\documentclass[11pt,a4paper]{amsart}
\usepackage[foot]{amsaddr}
\usepackage{ifxetex}
\ifxetex
  \usepackage[no-math]{fontspec}
\else
\fi
\usepackage{amsmath, amsfonts, amssymb, amsthm}
\usepackage{fullpage}
\usepackage{microtype}
\usepackage[libertine]{newtxmath}
\usepackage[tt=false]{libertine}
\usepackage{caption}
\usepackage{bbm}
\usepackage{hyperref, color}
\hypersetup{colorlinks=true,citecolor=blue, linkcolor=blue, urlcolor=blue}
\usepackage[linesnumbered,boxed,ruled,vlined]{algorithm2e}
\usepackage{bm}
\usepackage{bbm}
\usepackage[numbers]{natbib}
\usepackage{xcolor}
\usepackage{enumerate} 
\usepackage{enumitem}
\usepackage{tabularx}
\usepackage{array}
\usepackage{cleveref}
\newcolumntype{L}[1]{>{\raggedright\arraybackslash}p{#1}}
\newcolumntype{C}[1]{>{\centering\arraybackslash}m{#1}}
\newcolumntype{R}[1]{>{\raggedleft\arraybackslash}p{#1}}
  
\usepackage{makecell}
\usepackage{color}
\usepackage{graphicx}
\usepackage{float}
\usepackage{bbm}

\usepackage{footnote}
\makesavenoteenv{tabular}

\usepackage{thm-restate}
\usepackage{commath}

\newtheorem{theorem}{Theorem}[section]

\newtheorem{lemma}[theorem]{Lemma}

\newtheorem{remark}[theorem]{Remark}
\newtheorem{proposition}[theorem]{Proposition}

\usepackage{comment}

\usepackage{multicol, float}
\usepackage[linesnumbered,boxed,ruled,vlined]{algorithm2e}

\usepackage{thmtools, thm-restate, hyperref, cleveref}

\usepackage{mathtools}

\usepackage{enumitem}

\title{A Simple Distributed Algorithm for Sparse Fractional Covering and Packing Problems}

\date{}

\author{Qian Li}
\author{Minghui Ouyang}
\author{Yuyi Wang}

\address[Qian Li]{Shenzhen International  Center For Industrial  And  Applied  Mathematics, Shenzhen Research Institute of Big Data, China 
\textnormal{E-mail: \url{liqian.ict@gmail.com}}.}
\address[Minghui Ouyang]{School of Mathematical Sciences, Peking University, China \textnormal{E-mail: \url{ouyangminghui1998@gmail.com}}.}
\address[Yuyi Wang]{Lambda Lab, CRRC Zhuzhou, China \textnormal{E-mail: \url{yuyiwang920@gmail.com}}.}

\begin{document}

\maketitle

\begin{abstract}
This paper presents a distributed algorithm in the CONGEST model that achieves a $(1+\epsilon)$-approximation for row-sparse fractional covering problems (RS-FCP) and the dual column-sparse fraction packing problems (CS-FPP). Compared with the best-known $(1+\epsilon)$-approximation CONGEST algorithm for RS-FCP/CS-FPP developed by Kuhn, Moscibroda, and Wattenhofer (SODA'06), our algorithm is not only much simpler but also significantly improves the dependency on $\epsilon$.
\end{abstract}

\section{Introduction}\label{sec:introduction}
A \emph{fractional covering} problem (FCP) and its dual \emph{fractional packing} problem (FPP) are positive linear programs (LP) of the canonical form:
\begin{flalign*}
&&\min_{\vec x} \quad &\vec{c}^T\vec{x} \quad\quad\mbox{(FCP, the primal LP)} & \max_{\vec y} \quad &\vec{b}^T\vec{y}\quad\quad\mbox{(FPP, the dual LP)}&& \\
&&\mbox{s.t.}\quad &\vec{A}\vec{x}\geq \vec{b} & \mbox{s.t.}\quad &\vec{A}^T\vec{y}\leq \vec{c}&&\\
&& &\vec{x}\geq 0 & &\vec{y}\geq 0,&&
\end{flalign*}

where all $A_{ij},b_i$, and $c_j$ are non-negative. 
This paper particularly focuses on $k$-row-sparse FCPs ($k$-RS-FCP) and $k$-column-sparse FPPs ($k$-CS-FPP). These are FCPs and FPPs in which the matrix $\vec{A}$ contains at most $k$ non-zero entries per row. They are still fairly general problems and can model a broad class of basic problems in combinatorial optimization, such as the fractional version of vertex cover, bounded-frequency weighted set cover, weighted $k$-uniform hypergraph matching, stochastic matching, and stochastic $k$-set packing.

This paper studies distributed algorithms for FCPs and FPPs in the CONGEST model. The CONGEST model features a network $G=(V, E)$, where each node corresponds to a processor and each edge $(u,v)$ represents a bidirectional communication channel between processors $u$ and $v$. The computation proceeds in rounds. In one round, each processor first executes local computations and then sends messages to its neighbors. Each of the messages is restricted to $O(\log |V|)$ bits. The algorithm complexity is measured in the number of rounds it performs. 

For an FCP/FPP instance with $m$ primal variables and $n$ dual variables (i.e., $A$ has dimensions $n\times m$), the network is a bipartite graph $G=([m],[n], E)$. 
Each primal variable $x_j$ is associated with a left node $j \in [m]$, and each dual variable $y_i$ is associated with a right node $i \in [n]$. An edge $(i, j)$ exists if and only if $A_{ij}>0$. At the beginning of a CONGEST algorithm, each left node $j$ only knows its corresponding cost $c_j$ and the column vector $(A_{ij}: i \in [n])$, and each right node $i$ only knows $b_i$ and the row vector $(A_{ij}: j \in [m])$. At the end of the algorithm, each left node $j$ is required to output a number $\hat{x}_j$ and each right node $i$ to output a $\hat{y}_i$, which together are supposed to form approximate solutions to the FCP and FPP respectively. 

Let $\vec{1}_k$ denote the $k$-dimensional all-ones vector. The subscript $k$ will be dropped if it is implicit. Without loss of generality, this paper considers FCP and FPP instances of the following normal forms: 
\begin{equation}\label{normalfcp}
\min_{\vec x} \quad \vec{1}^T\vec{x}\quad \mbox{s.t.}\quad \vec{A}\vec{x}\geq \vec{1} \mbox{ and } \vec{x}\geq 0
\end{equation}
and 
\begin{equation}\label{normalfpp}
\max_{\vec y} \quad \vec{1}^T\vec{y}\quad \mbox{s.t.}\quad \vec{A}^T\vec{y}\leq \vec{1} \mbox{ and } \vec{y}\geq 0
\end{equation}
where $A_{ij}$ is either 0 or $\geq 1$. The reduction to the normal form proceeds as follows:
\begin{itemize}[leftmargin=24pt]
\item First, we can assume that each $b_i>0$ since otherwise we can set $y_i$ to zero and delete the $i$-th row of $A$; 

similarly, we can assume each $c_j>0$ since otherwise we can set $x_j$ to $+\infty$ and delete the $j$-th column of $A$, and then set $y_i$ to zero and delete the $i$-th row for all $i$ with $A_{ij}>0$. 
\item Then, we can replace $A_{ij}$ by $\hat{A}_{ij}=\frac{A_{ij}}{b_ic_j}$, replace $\vec{b}$ and $\vec{c}$ by all-ones vector, and work with variables $\hat{x}_j=c_jx_j$ and $\hat{y}_i=b_iy_i$. 
\item Finally, we replace $\hat{A}_{ij}$ with $\tilde{A}_{ij}=\frac{\hat{A}_{ij}}{\min\{\hat{A}_{i'j'}\mid \hat{A}_{i'j'> 0}\}}$ and work with $\tilde{x}_j=\hat{x}_j\cdot \min\{\hat{A}_{i'j'}\mid \hat{A}_{i'j'> 0}\}$ and $\tilde{y}_i=\hat{y}_i\cdot \min\{\hat{A}_{i'j'}\mid \hat{A}_{i'j'> 0}\}$.
\end{itemize}

Papadimitriou and Yannakakis \cite{papa} initiated the research on approximating FCPs/FPPs in the CONGEST model. Bartal, Byers, and Raz \cite{Bartal} proposed the first constant approximation ratio algorithm with $\mathrm{polylog}(m+n)$ rounds. After designing a distributed algorithm for a specific FCP/FPP scenario, namely the fractional  dominating set problem \cite{KuhnW05},  Kuhn, Moscibroda, and Wattenhofer \cite{soda06} finally developed an efficient $(1+\epsilon)$-approximation algorithm for general FCP/FPP instances, running in $O(\log \Gamma_p \cdot \log\Gamma_d /\epsilon^4)$ round for normalized instances where
 \[
\Gamma_p \coloneq \max_j \sum_{i=1}^{n}A_{ij}\mbox{, and }\Gamma_d \coloneq \max_i \sum_{j=1}^m A_{ij}.
\]
Particularly, for RS-FCPs/CS-FPPs, the round complexity becomes $O\left(\log A_{\max}\cdot\log \Gamma_p/\epsilon^4\right)$. Later, Awerbuch and Khandekar \cite{stoc08} proposed another $(1+\epsilon)$-approximation algorithm for general normalized FCP/FPP instances running in $\tilde{O}(\log^2(nA_{\max})\log^2(nmA_{\max})/\epsilon^5)$ rounds, which has worse bound than \cite{soda06} but enjoys the features of simplicity and statelessness. 

Several works studied the lower bound for CONGEST algorithms to approximate linear programming. Bartal, Byers, and Raz \cite{Bartal} showed that $(1+\epsilon)$-approximation algorithms for general FCPs/FPPs require at least $\Omega(1/\epsilon)$ rounds. Kuhn, Moscibroda, and Wattenhofer proved that \cite{Kuhn2005,Kuhn16,DBLP:conf/podc/KuhnMW04,soda06,LenzenW10} no constant round, constant-factor approximation CONGEST algorithms exist for the LP relaxation of minimum vertex cover, minimum dominating set, or maximum matching in general graphs. Later, they improved this lower bound by showing that \cite{KuhnMW16} no $o(\sqrt{\log (m+n)/\log\log (m+n)})$ rounds CONGEST algorithms can constant-factor approximate the LP relaxation of minimum vertex cover, maximum matching, or by extension, the general binary RS-FCPs or CS-FPPs (i.e., $A_{ij}\in\{0,1\}$) as well.

In this paper, we propose a CONGEST algorithm (Algorithm \ref{alg:main}) for approximating general RS-FCP/CS-FPP instances, which is much simpler than the algorithm of \cite{soda06}. Moreover, our algorithm exhibits a worse dependency on $A_{\max}$ but improves the dependency on $\epsilon$. In particular, for the binary RS-FCPs or CS-FPPs, which include LP relaxations of many combinatorial problems such as minimum vertex cover, minimum dominating set, maximum matching, and maximum independent set, our algorithm runs in $O(\log\Gamma_p/\epsilon^2)$ rounds, which is a  $1/\epsilon^2$ factor improvement over the algorithm of \cite{soda06}.
\begin{restatable}[Main Theorem]{theorem}{maintheorem}
\label{thm:main_thm}
For any $\epsilon>0$, Algorithm \ref{alg:main} computes $(1+\epsilon)$-approximate solutions to RS-FCP and CS-FPP at the same time, running in $O(A_{\max}\cdot \log\Gamma_p/\epsilon^2)$ rounds.
\end{restatable}

\begin{remark}
In 2018, Ahmadi et al. \cite{disc18} proposed a simple $(1+\epsilon)$-approximation distributed algorithm for the LP relaxation of minimum vertex cover and maximum weighted matching problems, which are special classes of 2-RS-FCPs and 2-CS-FPPs. They claimed the algorithm runs in $O(\log \Gamma_p/\epsilon^2)$ rounds, an $O(A_{\max})$ factor faster than our algorithm. Unfortunately, there is a flaw in their proof\footnote{In the proof of Lemma 5.2 in the full version, $Y_v^+$ should be defined as $y_e^+/w_e$ rather than $y_e^+$; $\frac{\alpha^{1/w_e}}{\alpha^{1/w_e}-1} \le \frac{\alpha}{\alpha-1}$ seems doubtful since $\frac{1}{w_e} > 1$ in their setting. }, and we do not know how to correct the proof to achieve the claimed bound.
\end{remark}

\section{Algorithms}

In this section, we present our algorithm (Algorithm \ref{alg:main}) and the analysis. Indeed, our algorithm applies to general FCP/FPP instances, and we will prove the following theorem: 

\begin{theorem} \label{thm:general_main_thm}
For any $\epsilon>0$, Algorithm \ref{alg:main} computes $(1+\epsilon)$-approximate solutions to \eqref{normalfcp} and \eqref{normalfpp} at the same time, running in $O \left(\Gamma_d \cdot \log \Gamma_p/\epsilon^2\right)$ rounds.
\end{theorem}

\begin{algorithm}[t]
	\caption{An $(1+\epsilon)$-approximation algorithm for the normalized FCP/FPP}
	\label{alg:main}
	\textbf{Parameter:} $\alpha,f\in\mathbb{R}_{\geq 0}$ and $L\in\mathbb{N}$ defined as in \eqref{alpha+f+L};\\
     Initialize $x_S:=0$ for any $S$, and $y_e:=0$ and $r_e:=1$ for any $e$;\\
     \For{$\ell=1$ to $L$}{
         \For{each all $S$ in parallel}{
              \If{$\rho_S\geq \frac{1}{\alpha} \cdot \max_{S'\cap S\neq\emptyset} \rho_{S'}$}{
                    $x_S:=x_S+1$;\\
                    \For{all $e\in S$}{
                          $y_e:=y_e+A_{eS}\cdot r_e/\left(\sum_{e\in S}A_{eS}\cdot r_e\right)$ and $r_e:=r_e/\alpha^{A_{eS}}$;\\
                          \lIf{$r_e\leq \alpha^{-f}$}{$r_e:=0$}
                    }
               }
         }
     }
    \textbf{Return:} $\vec{x}/f$ and $\vec{y}/(f\cdot (1+\epsilon))$ as the approximate solutions to \eqref{setcover} and \eqref{setcover-dual} respectively.
\end{algorithm}

Our algorithm is based on the sequential fractional set cover algorithm by Eisenbrand et al. \cite{DBLP:conf/soda/EisenbrandFGK03} and the fractional weighted bipartite matching by Ahmadi et al. \cite{disc18}. It will be helpful to view \eqref{normalfcp} as a generalization of the fractional set cover problem. Specficially, there is a universe $U=\{e_1,\cdots,e_n\}$ of $n$ elements, a collection $\mathcal{S}=\{S_1,S_2,\cdots,S_m\}$ of subsets of $U$, and a matrix $\{A_{eS}:e\in U,S\in\mathcal{S}\}$ indicating the covering efficiency of $S$ on $e$. We say $e\in S$ if $A_{eS}>0$. Then \eqref{normalfcp} can be recast as the following generalization of the fractional set cover problem:
\begin{equation}\label{setcover}
\min_{\vec{x}\geq 0} \quad \sum_{S\in \mathcal{S}} x_S\quad \mbox{s.t.}\quad \sum_{S\ni e} A_{eS}\cdot x_S\geq 1 \mbox{ for any $e$}.
\end{equation}
The dual \eqref{normalfpp} can be recast as:
\begin{equation}\label{setcover-dual}
\max_{\vec{y}\geq 0} \quad \sum_{e\in U} y_e\quad \mbox{s.t.}\quad \sum_{e\in S} A_{eS}\cdot y_e\leq 1 \mbox{ for any $S$}.
\end{equation}

Our algorithm maintains a variable $x_S$ for each subset $S$, and two variables $y_e$ and $r_e$ for each element $e$. $x_S$ and $y_e$ are initially $0$ and their values can only increase throughout the algorithm; the variable $r_e$ is initially $1$ and its value can only decreases. Intuitively, $r_e$ denotes the ``requirement'' of element $e$. Furthermore, we define the efficiency of $S$ as $\rho_S:=\sum_{e\in S} A_{eS}\cdot r_e$.

Our algorithm consists of $L$ phases. $\alpha$ and $f$ are two other algorithmic parameters. The values of $L,\alpha,f$ will be determined later. In the $\ell$-th phase,  the algorithm picks all subsets $S$ with 
\[
\rho_S\geq \frac{1}{\alpha} \cdot \max_{S'\cap S\neq\emptyset} \rho_{S}'.
\]
and update the primal variable
\[
x_S:=x_S+1,
\]
as well as: for each $e\in S$,
\[
y_e:=y_e+\frac{A_{eS}\cdot r_e}{\sum_{e\in S}A_{eS}\cdot r_e}, \mbox{ and } r_e:=r_e/\alpha^{A_{eS}}.
\]

In other words, let $\Xi_\ell (e):=\{S\ni e\mid S$ is selected in the $\ell$-th phase$\}$, then after the $\ell$-th phase, we have
\[
y_e:=y_e+\Delta y_e=y_e+\sum_{S\in \Xi_\ell(e)}\frac{A_{eS}\cdot r_e}{\sum_{e\in S}A_{eS}\cdot r_e}, \mbox{ and } r_e:=r_e/\alpha^{\sum_{S\in \Xi_\ell (e)} A_{eS}}.
\]
Besides, we set $r_e=0$ as soon as $r_e\leq \alpha^{-f}$. Finally, the algorithm returns $\vec{x}/f$ and $\vec{y}/\left((1+\epsilon)\cdot f\right)$ as the approximation solutions. See Algorithm \ref{alg:main} for a formal description.
\begin{remark}
	The two algorithms in \cite{soda06} and \cite{disc18} both have a similar greedy fashion: it starts with all $x_S$ set to $0$, always increases the $x_S$ whose ``efficiency'' is maximum up to a certain factor, and then distributes the increment of $x_S$ among its elements and decreases the requirements $r_e$. Our algorithm and the two algorithms of \cite{soda06} and \cite{disc18} differ in specific implementations: the definition of efficiency, the distribution of increments, and the reduction of requirements. In particular, the algorithm of \cite{soda06} consists of two levels of loops: the goal of the first-level loop is to reduce the maximum ``weighted primal degree'', and one complete run of the second-level loop can be seen as one parallel greedy step. This two-level structure complicates the algorithm of \cite{soda06}. The algorithm of \cite{disc18} only works for the LP relaxation of minimum vertex cover and maximum weighted matching problems, which are special classes of 2-RS-FCPs and 2-CS-FPPs. The distribution of increments and the reduction of requirements in our implementation is similar to \cite{disc18}, and the main difference is the definition of efficiency.  
\end{remark}

Before analyzing the correctness and efficiency of the algorithm, we present some helpful observations about its behavior.
\begin{proposition}\label{observations}
Throughout the algorithm, we always have 
\begin{itemize}[leftmargin=24pt]
\item[(a)] $\sum_{S} x_S = \sum_e y_e$.
\item[(b)] For any $S$, the value of $\rho_S$ is non-increasing, and lies within $\left(\alpha^{-f}, \Gamma_p\right] \cup \{0\}$.
\item[(c)] After each phase, $\max_{S}\rho_S$ decreases by a factor of at least $\alpha$. 
\end{itemize}
\end{proposition}
\begin{proof}
\textbf{Part (a).} Initially, $\sum_{S} x_S = \sum_e y_e=0$. Then, whenever we increase $x_S$ by $1$, we increase $\sum_{e} y_e$ by 
$\sum_{e\in S} \frac{A_{eS}\cdot r_e}{\sum_{e\in S}A_{eS}\cdot r_e}=1$.

\textbf{Part (b).} Since $r_e$ is non-increasing, so is $\rho_S:=\sum_{e\in S} A_{eS}\cdot r_e$. Besides, the initial value of $\rho_S$ is $\sum_{e\in S} A_{eS}$, which is upper-bounded by $\Gamma_p$. Furthermore, for any non-zero $r_e$ where $e\in S$, it should be strictly greater than $\alpha^{-f}$, since otherwise it will be set to 0. Recalling that $A_{eS} \geq 1$, we have $\rho_S>\alpha^{-f}$ or $\rho_S=0$. 

\textbf{Part (c).} Define $\rho_{\max}:=\max_{S}\rho_S$. Note that every $S$ with $\rho_S\geq \rho_{\max}/\alpha$ will be picked. Then for any such $S$, $r_e$ will decrease by a factor of $\alpha^{\sum_{S'\in \Xi_\ell (e)} A_{eS}}\geq \alpha^{A_{eS}}\geq \alpha$ for any $e\in S$, so $\rho_S=\sum_{e\in S} A_{eS}\cdot r_e$ decreases by a factor of at least $\alpha$ as well. 
\end{proof}

By Proposition \ref{observations} (b) and (c), it is easy to see that after $\lceil\log_{\alpha} \Gamma_p+f\rceil$ phases, all $\rho_S$ and $r_e$ will become zero. We choose \footnote{The reason behind the choices of parameters is that $\alpha$ should sufficiently close to $1$ and $f$ should sufficiently large, such that $\alpha^{\Gamma_d+1} = 1 + O(\epsilon)$ and $\frac{\ln \Gamma_p}{\ln \alpha} \ll f$. See the proof of \Cref{thm:correctness} for details. }
\begin{equation}\label{alpha+f+L}
\alpha:=1+\frac{\epsilon}{c\cdot \Gamma_d}, \mbox{ and } f:=\frac{2}{\epsilon\cdot \ln \alpha}\cdot \ln \Gamma_p, \mbox{ and } 
L:=\lceil\log_{\alpha} \Gamma_p+f\rceil.
\end{equation}
where $c$ is a sufficiently large constant. 
Note that each phase can be implemented in constant rounds. So the following lemma holds.
\begin{lemma}\label{lem:complexity}
Algorithm \ref{alg:main} runs in $O(\Gamma_d\cdot \log \Gamma_p/\epsilon^2)$ rounds. When it terminates, all $r_e=0$ and all $\rho_S=0$.
\end{lemma}
What remains is to prove its correctness. Let $\vec{x}^L$ and $\vec{y}^L$ denote the values of $\vec{x}$ and $\vec{y}$ right after the $L$-th phase. We first prove the feasibility.
\begin{theorem} \label{thm:correctness}
    $\vec{x}^L/f$ and $\vec{y}^L/\left((1+\epsilon)\cdot f\right)$ are feasible solutions to \eqref{setcover} and \eqref{setcover-dual} respectively.
\end{theorem}
\begin{proof}
We first show the feasibility of $\vec{x}^L/f$. Obviously, $\vec{x}^L/f$ are non-negative. Given any $e$, whenever we increase $x_S$ by 1 for some $S\ni e$, we divide $r_e$ by a factor $\alpha^{A_{eS}}$. The initial value of $r_e$ is $1$; and by Lemma \ref{lem:complexity}, finally $r_e$ becomes $\leq \alpha^{-f}$, and then is set to $0$. We therefore have
\begin{equation*} 
\sum_{S\ni e} A_{eS}\cdot x_S^L\geq f,
\end{equation*}
and then conclude the feasibility of $\vec{x}^L/f$.

In the following, we prove the feasibility of $\vec{y}^L/f$. Obviously, $\vec{y}^L/((1+\epsilon)\cdot f)$ are non-negative. What remains is to show that for any $S$
\[
\sum_{e\in S} A_{eS}\cdot y_e^L \leq (1+\epsilon)\cdot f.
\]
For the convenience of presentation, define $Y_S:=\sum_{e\in S} A_{eS}\cdot y_e$, which only increases during the algorithm's execution. The idea is to upper bound the increment $\Delta Y_S$ of $Y_S$ in terms of the decrement $\Delta \rho_S$ of $\rho_S$ in each phase. 

On the one hand, recalling that the increment of $y_e$ is 
\[
\Delta y_e=\sum_{S'\in \Xi_\ell(e)}\frac{A_{eS'}\cdot r_e}{\sum_{e\in S'}A_{eS'}\cdot r_e}=\sum_{S'\in \Xi_\ell(e)} \frac{1}{\rho_{S'}}\cdot A_{eS'}\cdot r_e,
\]
we have
\begin{align*}
\Delta Y_S=&\sum_{e\in S} A_{eS}\cdot \Delta y_e=\sum_{e\in S} \sum_{S'\in \Xi_\ell(e)}\frac{1}{\rho_{S'}}\cdot A_{eS}\cdot A_{eS'}\cdot r_e\leq \sum_{e\in S} \sum_{S'\in \Xi_\ell(e)}\frac{\alpha}{\rho_{S}}\cdot A_{eS}\cdot A_{eS'}\cdot r_e\\
= & \frac{\alpha}{\rho_{S}}\sum_{e\in S} \sum_{S'\in \Xi_\ell(e)} A_{eS}\cdot A_{eS'}\cdot r_e.
\end{align*}

On the other hand, 
\begin{align*}
\Delta\rho_S=\sum_{e\in S} A_{eS}\cdot \Delta r_e=&\sum_{e\in S} A_{eS}\cdot r_e\cdot \left(1-1/\alpha^{\sum_{S'\in \Xi_\ell(e)}A_{eS'}}\right)\\
=&\sum_{e\in S} A_{eS}\cdot r_e\cdot \frac{1}{\alpha^{\sum_{S'\in \Xi_\ell(e)}A_{eS'}}}\left(\alpha^{\sum_{S'\in \Xi_\ell(e)}A_{eS'}}-1\right)\\
\geq &\sum_{e\in S} A_{eS}\cdot r_e\cdot \frac{1}{\alpha^{\sum_{S'\in \Xi_\ell(e)}A_{eS'}}}\left(\sum_{S'\in \Xi_\ell(e)}\ln\alpha\cdot  A_{eS'}\right)\\
\geq& \frac{\ln \alpha}{\alpha^{\Gamma_d}}\sum_{e\in S}\sum_{S'\in \Xi_\ell(e)}A_{eS}\cdot A_{eS'}\cdot r_e
    \end{align*}

    Combining the two inequalities above, we get
    \[ \Delta Y_S \leq \frac{\alpha^{\Gamma_d + 1}}{\ln \alpha} \cdot \frac{\Delta \rho_S}{\rho_S}.\]
   Then, summing this inequality up over all phases, we have
    \begin{align*}
        Y^{\text{end}}_S &= \sum_{\text{each phase}} \Delta Y_S \le \frac{\alpha^{\Gamma_d + 1}}{\ln \alpha} \sum_{\text{each phase}} \frac{\Delta \rho_S}{\rho_S} \le \frac{\alpha^{\Gamma_d + 1}}{\ln \alpha} \int_{\rho_S^{\text{end}}}^{\rho_S^{\text{initial}}} \frac{1}{\rho} \, \text{d} \rho \\
        &= \frac{\alpha^{\Gamma_d + 1}}{\ln \alpha} \cdot \left( \ln \rho_S^{\text{initial}} - \ln \rho_S^{\text{end}} \right)
        \le \frac{\alpha^{\Gamma_d + 1}}{\ln \alpha} \cdot \left( \ln \alpha \cdot f + \ln \Gamma_p \right) \\
        &= \alpha^{\Gamma_d + 1} f + \frac{\alpha^{\Gamma_d + 1} \ln \Gamma_p}{\ln \alpha} \le (1+\epsilon) f. \qedhere
    \end{align*}
\end{proof}
By Proposition \ref{observations} (a), we have $\sum_S x_S^L=\sum_e y_e^L$, which means
\[
\sum_e y_e^L/\left((1+\epsilon)\cdot f\right)\leq (1+\epsilon) \sum_S x_S^L/f.
\] 
We therefore conclude that  $\vec{x}^L/f$ and $\vec{y}^L/\left((1+\epsilon)\cdot f\right)$ are $(1+\epsilon)$-approximate solutions to \eqref{setcover} and \eqref{setcover-dual} respectively. By putting it and Lemma \ref{lem:complexity} together, we finish the proof of Theorem \ref{thm:general_main_thm}.

\section{Conclusion}
This paper proposes a simple $(1+\epsilon)$-approximation CONGEST algorithm for row-sparse fractional covering problems and column-sparse fractional packing problems. It runs in $O \left(A_{\max}\cdot \log \Gamma_p/\epsilon^2\right)$ rounds, where $\Gamma_p=\max_{j}\sum_i A_{ij}$ and $\Gamma_d=\max_{i}\sum_j A_{ij}$. Our algorithm is simpler than the algorithm of \cite{soda06}, worsens the $A_{max}$-dependency, but improves the $\epsilon$-dependency. 
For future work, 
it is an intriguing open problem, proposed by Suomela \cite{survey}, whether constant round, constant-factor approximation CONGEST algorithms exist for row-sparse, column-sparse FCP/FPP instances—a special kind of RS-FCP/CS-FPP where the number of nonzero entries in each column of $A$ is also bounded. Our algorithm and the algorithm of \cite{soda06} are both such algorithms for instances where $A_{\max}$ is bounded. 

\section*{Acknowledgement}
Qian Li's work is supported by Hetao Shenzhen-Hong Kong Science and Technology Innovation Cooperation Zone Project
(No.HZQSWS-KCCYB-2024016), and the National Natural Science Foundation of China Grants No.6200222.

\bibliographystyle{abbrv}
\bibliography{Distributed_LP_algorithms}

\end{document}